\documentclass[11pt,a4paper]{article}
\pdfoutput=1

\usepackage[utf8]{inputenc}
\usepackage[T1]{fontenc}
\usepackage[english]{babel}
\usepackage{verbatim}

\usepackage{tikz}
\usepackage{amssymb}
\usepackage{mathtools}
\usepackage{dsfont}
\usepackage{bbm}
\usepackage{mleftright}\mleftright
\usepackage{amsthm}
\usepackage{thmtools}
\usepackage{thm-restate}

\usepackage{algorithm}
\usepackage{algcompatible}
\usepackage{enumitem}	
\usepackage{float}

\usepackage{caption}

\usepackage{diagbox}

\usepackage{footnote}
\newcommand{\footref}[1]{\textsuperscript{\textup{#1}}}

\newcommand{\eps}{\epsilon}
\newcommand{\bigket}[1]{\left |#1 \right \rangle}

\newcommand{\ketbra}[2]{|#1\rangle\! \langle #2|}

\newcommand{\nrm}[1]{\left\lVert #1 \right\rVert}
\newcommand{\bigO}[1]{\mathcal{O}\left( #1 \right)}
\newcommand{\bOt}[1]{\widetilde{\mathcal{O}}\left( #1 \right)}

\newcommand{\range}[1]{[#1]}

\newcommand{\img}[1]{\mathrm{img}\left( #1 \right)}

\newcommand{\specialcell}[2][c]{%
  \begin{tabular}[#1]{@{}c@{}}#2\end{tabular}}

\DeclareMathOperator{\poly}{poly}

\DeclarePairedDelimiter\bra{\langle}{\rvert}
\DeclarePairedDelimiter\ket{\lvert}{\rangle}
\DeclarePairedDelimiterX\braket[2]{\langle}{\rangle}{#1 \delimsize\vert #2}

\def\Tr{\mathrm{Tr}}

\def\>{\rangle}
\def\<{\langle}

\providecommand{\tr}[1]{\Tr\left[#1\right]}

\providecommand{\eend}[1]{\mathrm{End}\left(#1\right)}

\long\def\ignore#1{}

\newtheorem{theorem}{Theorem}
\newtheorem{corollary}[theorem]{Corollary}
\newtheorem{lemma}[theorem]{Lemma}
\newtheorem{definition}[theorem]{Definition}

\newcommand{\eq}[1]{(\ref{eq:#1})}
\newcommand{\thm}[1]{\hyperref[thm:#1]{Theorem~\ref*{thm:#1}}}
\newcommand{\cor}[1]{\hyperref[cor:#1]{Corollary~\ref*{cor:#1}}}
\newcommand{\defn}[1]{\hyperref[def:#1]{Definition~\ref*{def:#1}}}
\newcommand{\lem}[1]{\hyperref[lem:#1]{Lemma~\ref*{lem:#1}}}
\newcommand{\prop}[1]{\hyperref[prop:#1]{Proposition~\ref*{prop:#1}}}
\newcommand{\fig}[1]{\hyperref[fig:#1]{Figure~\ref*{fig:#1}}}
\newcommand{\tab}[1]{\hyperref[tab:#1]{Table~\ref*{tab:#1}}}
\newcommand{\alg}[1]{\hyperref[alg:#1]{Algorithm~\ref*{alg:#1}}}
\renewcommand{\sec}[1]{\hyperref[sec:#1]{Section~\ref*{sec:#1}}}
\newcommand{\subsec}[1]{\hyperref[subsec:#1]{Section~\ref*{subsec:#1}}}
\newcommand{\apx}[1]{\hyperref[apx:#1]{Appendix~\ref*{apx:#1}}}
\newcommand{\fac}[1]{\hyperref[fac:#1]{Fact~\ref*{fac:#1}}}

\newcommand{\A}{\ensuremath{\mathcal{A}}}

\newcommand{\N}{\ensuremath{\mathbb{N}}}

\newcommand{\R}{\ensuremath{\mathbb{R}}}

\newcommand{\Ex}{\mathbb{E}}

\newcommand{\C}{\mathbb{C}}
\renewcommand{\H}{\mathcal{H}}
\usepackage{tikz}

\usetikzlibrary{backgrounds,fit,decorations.pathreplacing,calc}

\usepackage{titling}
\usepackage{fullpage}

\usepackage[final]{hyperref}
\hypersetup{
	colorlinks = true,
	allcolors = {blue},
}

\usepackage{ifdraft}
\ifdraft{\newcommand{\authnote}[3]{{\color{#3} {\bf  #1:} #2}}}{\newcommand{\authnote}[3]{}}

\usepackage{natbib}
\setcitestyle{authoryear,open={(},close={)}}
\newcommand*{\doi}[1]{\href{http://doi.org/#1}{doi: #1}}

\title{Distributional property testing in a quantum world}
	\author{
	András Gilyén\thanks{QuSoft, CWI and University of Amsterdam, the Netherlands. Supported by ERC Consolidator Grant QPROGRESS and partially supported by QuantERA project QuantAlgo 680-91-034. {\tt gilyen@cwi.nl}} 
	\and
	Tongyang Li\thanks{Department of Computer Science, Institute for Advanced Computer Studies, and Joint Center for Quantum Information and Computer Science, University of Maryland. Supported by IBM PhD Fellowship, QISE-NET Triplet Award (NSF DMR-1747426), and the U.S. Department of Energy, Office of Science, Office of Advanced Scientific Computing Research, Quantum Algorithms Teams program. {\tt tongyang@cs.umd.edu}}
}
\long\def\ignore#1{#1}
\newcommand{\cnewpage}{}
\newcommand{\cvspace}[1]{}
\newcommand{\cvskip}[1]{}

\begin{document}
\maketitle

\begin{abstract}%
A fundamental problem in statistics and learning theory is to test properties of distributions. We show that quantum computers can solve such problems with significant speed-ups. In particular, we give fast quantum algorithms for testing closeness between unknown distributions, testing independence between two distributions, and estimating the Shannon / von Neumann entropy of distributions. The distributions can be either classical or quantum, however our quantum algorithms require coherent quantum access to a process preparing the samples. Our results build on the recent technique of quantum singular value transformation, combined with more standard tricks such as divide-and-conquer. The presented approach is a natural fit for distributional property testing both in the classical and the quantum case, demonstrating the first speed-ups for testing properties of density operators that can be accessed coherently rather than only via sampling; for classical distributions our algorithms significantly improve the precision dependence of some earlier results.
\end{abstract}


\section{Introduction}\label{sec:intro}
Distributional property testing is a fundamental problem in theoretical computer science (see, e.g. \cite{goldreich2017introduction}). In such property testing questions the goal is to determine properties of probability distributions with the least number of independent samples. This has intimate connections and applications to statistics, learning theory, and algorithm design.

The merit of distributional property testing mainly comes from the fact that the testing of many properties admits \emph{sublinear} algorithms. For instance, given the ability to take samples from a discrete distribution $p$ on $\range{n}:=\{1,\ldots,n\}$, it requires $\Theta(n/\epsilon^{2})$ samples to ``learn'' $p$, i.e., to construct a distribution $q$ on $\range{n}$ such that $\|p-q\|_{1}\leq\epsilon$ with success probability at least $2/3$ ($\|\cdot\|_{1}$ being $\ell^{1}$-distance). However, testing whether $p=q$ or $\|p-q\|_{1}>\epsilon$ requires only $\Theta(\max\{\frac{n^{2/3}}{\epsilon^{4/3}},\frac{n^{1/2}}{\epsilon^{2}}\})$ samples from $p$ and $q$ (\cite{chan2014OptimalAlgTestCloseDistr}), which is sublinear in $n$ and significantly smaller than the complexity of learning the entire distributions. See \sec{related-works} for more examples and discussions.

In this paper, we study the impact of quantum computation on distributional property testing problems. We are motivated by the emerging topic of ``quantum property testing'' (see the survey of \cite{montanaro2013SurveQuantPropTest}) which focuses on investigating the quantum advantage in testing classical statistical properties. Quantum speed-ups have already been established for a few specific problems such as testing closeness between distributions (\cite{bravyi2011QAlgTestingPropOfDistr,montanaro2015QMonteCarlo}), testing identity to known distributions (\cite{chakraborty2010NewResQuantPropTest}), estimating entropies (\cite{li2017QQueryEntropyComp}), etc. In this paper we propose a generic approach for quantum distributional property testing, and illustrate its power on a few examples. This is our attempt to make progress on the question:
\begin{quote}
	\hspace{-5mm}\emph{Can quantum computers test properties of distributions systematically and more efficiently?}\kern-5mm
\end{quote}

\subsection{Problem statements}
Throughout the paper, we denote probability distributions on $\range{n}$ by $p$ and $q$; their $\ell^{\alpha}$-distance is defined as $\nrm{p-q}_{\alpha}:=(\sum_{i=1}^{n}|p_{i}-q_{i}|^{\alpha})^{\frac{1}{\alpha}}$. Similarly, we denote $n\times n$ density operators\footnote{For readers less familiar with quantum computing, a density operator (=quantum distribution) $\rho\in\C^{n\times n}$ is a positive semidefinite matrix with $\Tr[\rho]=1$. Please refer to the textbook \cite{nielsen2002QCQI} for more information.} ($=$quantum distributions) by $\rho$ and $\sigma$; their $\ell^{\alpha}$-distance is defined via the corresponding Schatten norm. 

\begin{paragraph}{Input models.}
To formulate the problems we address, we define classical and quantum access models for distributions on $[n]$. We begin with the very natural model of sampling.
\begin{definition}[Sampling]
A classical distribution $(p_{i})_{i=1}^{n}$ is accessible via \emph{classical sampling} if we can request samples from the distribution, i.e., get a random $i\in [n]$ with probability $p_{i}$. A quantum distribution $\rho\in\C^{n\times n}$ is accessible via \emph{quantum sampling} if we can request copies of the state $\rho$.
\end{definition}

Now we define a coherent analogue of the above sampling model. To our knowledge this type of query-access was not studied before in detail, especially in the context of density operator testing. 
The motivation for this input model is the following: we can think about a density operator as the outcome of some physical process. If we are able to simulate the corresponding process on a fault-tolerant quantum computer, then it provides purified access to the density operator. In the special case when we study a classical probability distribution coming from some classical randomized process, we can simply simulate the classical randomized process on a quantum computer.

\begin{definition}[Purified quantum query-access]\label{def:purified-quantum-query}
A density operator $\rho\in\C^{n\times n}$, has \emph{purified quantum query-access} if we have access to a unitary oracle $U_{\rho}$ (and its inverse) acting as\footnote{$|\psi\>\in\C^{n}$ denotes a ``ket'' vector and $\<\psi|=(|\psi\>)^{\dagger}$ stands for its conjugate transpose, called ``bra'' in Dirac notation; $|i\>\!=\!\vec{e}_{i}$ is the $i^{\text{th}}$ basis vector. An $\ell^2$-normalized $|\psi\>$ is called a pure state, and corresponds to density operator $\ketbra{\psi}{\psi}$. For $A=\C^k,B=\C^n$ and $\ket{\phi}\in A \otimes B$ we denote by $\tr{\ketbra{\phi}{\phi}}_A\in B\otimes B^* = \C^{n\times n}$ the partial trace over $A$.}
\cvskip{-5.5mm}
\begin{equation*}
U_{\rho}\ket{0}_A\ket{0}_B=\ket{\psi_\rho}_{AB}=\sum_{i=1}^{n}\sqrt{p_i}\ket{\phi_i}_A\ket{\psi_i}_B,
\text{ (where } \braket{\phi_i}{\phi_j}=\braket{\psi_i}{\psi_j}= \text{(Kronecker) } \delta_{ij})
\end{equation*}
\cvskip{-1.5mm} \noindent
such that $\Tr_A\left(\ketbra{\psi_\rho}{\psi_\rho}\right)=\rho$. 
If $\ket{\psi_i}=\ket{i}$, then $\rho=\sum_{i}^{n}p_i\ketbra{i}{i}$ is a diagonal density operator which can be identified with the classical distribution $p$, so we can simply write $U_{p}$ instead of $U_{\rho}$.
With a slight abuse of notation sometimes we will concisely write $\ket{\rho}$ instead of~$\ket{\psi_\rho}$.
\end{definition}

We also define an even stronger input model that is considered in a series of earlier works, see, e.g.,~(\cite{bravyi2011QAlgTestingPropOfDistr,chakraborty2010NewResQuantPropTest,li2017QQueryEntropyComp,bun2018PolyMethodStrikesBack}).

\cvspace{-0.7mm}
\begin{definition}[Classical distribution with discrete query-access]\label{def:discrete-query}
	$\!$A classical distribution $(p_{i})_{i=1}^{n}$, has \emph{discrete query-access} if we have classical / quantum query-access to a function $f\colon S \to [n]$ such that for all $i\in\range{n}$, $p_{i}=|\{s\in\range{S}: f(s)=i\}|/S$.
		(Typically the interesting regime is when $|S|\gg n$.)
	In the quantum case a query oracle is a unitary operator $\mathrm{O}$ acting on $\mathbb{C}^{|S|}\otimes \mathbb{C}^n$ as
	\cvskip{-3mm}
	\begin{equation*}
	\mathrm{O} \colon \ket{s,0}\leftrightarrow\ket{s,f(s)} \text{ for all }s\in S.
	\end{equation*}
\end{definition}
\cvskip{-1mm}
Note that if one first creates a uniform superposition over $S$ and then makes a query, then the above oracle turns into a purified query oracle to a classical distribution as in \defn{purified-quantum-query}. Therefore all lower bounds that are proven in this model also apply to the purified query-access oracles. In fact all algorithms that the authors are aware of do this conversion, so they effectively work in the purified query-access model. Moreover, we conjecture that the two input models are equivalent when $|S|\gg n$. For this reason we only work with the purified query-access model in this work.

\cnewpage

Another strengthening of the purified query-access model for classical distributions when we have access to a unitary (and its inverse) acting as $\ket{0}\mapsto\sum_{i=1}^{n}\sqrt{p_i}\ket{i}$.
\cvspace{-2mm}
\begin{definition}[Classical distribution with pure-state preparation access]\label{def:pureSTateQuery}
	$\!$A classical distribution $(p_{i})_{i=1}^{n}$, is accessible via \emph{pure state preparation oracle} if we have access to a unitary oracle $U_{\text{pure}}$ (and its inverse) acting as
	\cvskip{-8mm}
	$$
	U_{\text{pure}}\colon \ket{0}\mapsto\sum_{i=1}^{n}\sqrt{p_i}\ket{i}.
	$$
\end{definition}
\cvspace{-2mm}
	
This is again strictly stronger than the purified query-access model. In order to simulate purified queries we can first do a pure state query and then copy $\ket{i}$ to a second fresh ancillary register using, e.g., some CNOT gates. Finally, for completeness we mention that one could also consider a model similar to the above where one can only request samples of pure states of the form $\sum_{i=1}^{n}\sqrt{p_i}\ket{i}$, as studied for example in~\cite{arunachalam2017OptQSampCoplLearn,arunachalam2017TwoNewResQuantExactLearn}.

We will mostly focus on the first two input models and will only use the latter strengthenings of the purified query-access model for invoking and proving lower bounds.
\end{paragraph}

\begin{paragraph}{Property testing problems.}
We study three distributional properties: $\ell^{\alpha}$-closeness testing, independence testing, and entropy estimation. These properties are highly-representative; classically, these testers motivate general algorithms for testing properties of discrete distributions (\cite{diakonikolas2016NewTestPropDistr,acharya2017UnifiedEstimating}).

For brevity we only give the definitions for classical distributions; similar definitions apply to quantum density matrices if we replace vector norms by the corresponding Schatten norms.
\cvspace{-1mm}
\begin{definition}[$\ell^{\alpha}$-closeness testing]
Given $\epsilon>0$ and two probability distributions $p$, $q$ on $\range{n}$, \emph{$\ell^{\alpha}$-closeness testing} is to decide whether $p\!=\!q$ or $\nrm{p\!-\!q}_{\alpha}\geq\eps$ with success probability at least $\frac{2}{3}$.
Robust testing: decide whether $\!\nrm{p\!-\!q}_{\alpha}\!\leq\!0.99\eps$ or $\nrm{p\!-\!q}_{\alpha}\geq\eps$ with success probability at least $\frac{2}{3}$.
\end{definition}
\cvspace{-3mm}
\begin{definition}[Independence testing]
Given $\epsilon>0$ and a probability distribution $p$ on $\range{n}\times\range{m}$ with $n\geq m$, \emph{independence testing} is to decide, with success probability at least $\frac{2}{3}$, whether $p$ is a product distribution or $p$ is $\epsilon$-far in $\ell^{1}$-norm from any product distribution on $\range{n}\times\range{m}$ .
\end{definition}
\cvspace{-4mm}
\begin{definition}[Entropy estimation]
Given $\epsilon>0$ and a density operator $\rho\in\C^{n\times n}$, \emph{entropy estimation} is to estimate the Shannon / von Neumann entropy $H(\rho)=-\tr{\rho\log(\rho)}$ within additive $\epsilon$-precision with success probability at least $\frac{2}{3}$.
\end{definition}
\cvspace{-5mm}
\end{paragraph}

\subsection{Contributions}
We give a systematic study of distributional property testing for classical / quantum distributions, and obtain the following results for the purified quantum query model of \defn{purified-quantum-query}:
\cvspace{-2mm}
\begin{itemize}[leftmargin=4mm,itemsep=0mm]
	\item Entropy estimation of classical / quantum distributions costs  $\bOt{\frac{\sqrt{n}}{\epsilon^{1.5}}}$ and $\bOt{\frac{n}{\epsilon^{1.5}}}$
	queries respectively, as we prove in \thm{classical-entropy-purified} and \thm{quantum-entropy-purified}.	
	\item Robust $\ell^{2}$-closeness testing of classical / quantum distributions costs $\tilde{\Theta}\left(\frac{1}{\epsilon}\right)$ and $\mathcal{O}\kern-0.3mm\left(\kern-0.4mm\min\!\big(\frac{\sqrt{n}}{\epsilon}\!,\frac{1}{\epsilon^2}\big)\!\kern-0.1mm\right)\!$ queries respectively, as we prove in \thm{classical-l2-purified} and \thm{quantum-l2-purified}.					
	\item $\ell^{1}$-closeness testing of classical / quantum distributions costs $\bOt{\frac{\sqrt{n}}{\epsilon}}$ and $\bigO{\frac{n}{\epsilon}}$ queries respectively, as we prove in \cor{l1-purified}.
	\item Independence testing of classical / quantum distributions costs $\bOt{\frac{\sqrt{nm}}{\epsilon}}$ and $\bigO{\frac{nm}{\epsilon}}$ queries respectively, as we prove in \cor{independence-proof}.
\end{itemize}

\noindent
For context, we compare our results with previous classical and quantum results in \tab{results} and \tab{entropy}. (Note that all of our results are gate efficient, because they are based on singular value transformation and amplitude estimation, both of which have gate-efficient implementations.)

\begin{table}[H]
\Huge
\centering
\resizebox{1\columnwidth}{!}{%
\def\arraystretch{1.2}
\begin{tabular}{|c||c|c|c|}
\hline
\backslashbox{model}{problem} & $\ell^{1}$-closeness testing & (robust) $\ell^{2}$-closeness testing & Shannon / von Neumann entropy \\ \hline\hline
Classical sampling & \specialcell{$\Theta\left(\max\Big\{\frac{n^{2/3}}{\epsilon^{4/3}},\frac{n^{1/2}}{\epsilon^{2}}\Big\}\right)$ \\\cite{chan2014OptimalAlgTestCloseDistr}} & $\Theta\left(\frac{1}{\epsilon^{2}}\right)$ \cite{chan2014OptimalAlgTestCloseDistr}&
\specialcell{$\Theta\left(\frac{n}{\epsilon\log n}+\frac{\log^2 n}{\epsilon^{2}}\right)^{\phantom{\bar{a}}}\!\!$ \cite{jiao2015MinimaxEstFunDiscDistr}, \\\cite{wu2016MinimaxRatesEntrEstBestPolyApx}}\\ \hline
\specialcell{Classical with \\ quantum query-access} & $\boldsymbol{\bOt{\frac{\sqrt{n}}{\epsilon}}}$ & $\boldsymbol{\tilde{\Theta}\left(\frac{1}{\epsilon}\right)}$ &  $\boldsymbol{\bOt{\frac{\sqrt{n}}{\epsilon^{1.5}}}}$; $\tilde{\Omega}(\sqrt{n})$ \cite{bun2018PolyMethodStrikesBack} \\ \hline
\specialcell{Quantum state \\ with purification} & $\boldsymbol{\bigO{\frac{n}{\epsilon}}}$ & $\boldsymbol{\bigO{\min\left(\frac{\sqrt{n}}{\epsilon},\frac{1}{\eps^2}\right)}}$ & $\boldsymbol{\bOt{\frac{n}{\epsilon^{1.5}}}}$ \\ \hline
\specialcell{Quantum state\\ sampling} & $\Theta\left(\frac{n}{\epsilon^{2}}\right)$ \cite{badescu2017QStateCertification} & $\Theta\left(\frac{1}{\epsilon^{2}}\right)$ \cite{badescu2017QStateCertification} & $\bigO{\frac{n^{2}}{\epsilon^{2}}}$, $\Omega\Big(\frac{n^{2}}{\epsilon}\Big)$ \cite{acharya2017MeasuringQuantEntropy} \\ \hline
\end{tabular}
}
\stepcounter{footnote}
\caption{Summary of sample and query complexity results. Our new bounds are printed in \textbf{bold}. For classical distributions with quantum query-access\footref{\thefootnote} we prove (almost) matching upper and lower bounds for $\ell^2$-testing, and improve the previous best complexity $\bOt{\sqrt{n}/\epsilon^{2.5}}$ for $\ell^1$-testing by \cite{montanaro2015QMonteCarlo} and $\bOt{\sqrt{n}/\epsilon^{2}}$ for Shannon entropy estimation by \cite{li2017QQueryEntropyComp}. We~are not aware of prior work on testing quantum distributions with purified query-access.
}
\label{tab:results}
\end{table}
\footnotetext{Recent results of \cite{chailloux2018QueryPermutation} imply that in this model quantum speed-ups are at most cubic.}

\begin{table}[H]
\Huge
\centering
\cvskip{-3mm}
\resizebox{0.8\columnwidth}{!}{%
\def\arraystretch{2}
\begin{tabular}{|c||c|c|}
\hline
 & Sample complexity & (Purified) Query complexity \\ \hline\hline
Classical & $\Theta\left(\frac{n}{\log n}\right)$ \cite{valiant2011EstimatingTheUnseen} &  $\widetilde{\Theta}\left(\sqrt{n}\right)$ \cite{li2017QQueryEntropyComp,bun2018PolyMethodStrikesBack} \\ \hline
Quantum &  $\Theta\left(n^2\right)$ \cite{acharya2017MeasuringQuantEntropy} & $\boldsymbol{\bOt{n}}$ \\ \hline
\end{tabular}
}
\caption{Complexities of Shannon / von Neumann entropy estimation with constant precision. It seems that the $n$-dependence is roughly quadratically higher for quantum distributions, while coherent quantum access gives a quadratic advantage for both classical and quantum distributions. This suggests that our entropy estimation algorithm has essentially optimal $n$-dependence for density operators with purified access, however we do not have a matching lower bound yet.}
\label{tab:entropy}
\end{table}

\subsection{Techniques}\label{subsec:techniques}
The motivating idea behind our approach is that if we can prepare a purification of a quantum distribution / density operator $\rho$, then we can construct a unitary $U$, which has this density operator in the top-left corner, using only two queries to $U_\rho$. This observation is originally due to~\cite{low2016HamSimQubitization}. We call such a unitary a \emph{block-encoding} of $\rho$:
\cvskip{-3mm}
$$ U= \left[\begin{array}{cc} \rho & . \\ . & .    \end{array}\right] \Longleftrightarrow \rho= \left(\bra{0}^{\otimes a}\otimes I\right) U \left(\ket{0}^{\otimes a}\otimes I\right).$$
\cvskip{-1mm}\noindent
One can think of a block-encoding as a probabilistic implementation of the linear map $\rho$: given an input state $\ket{\psi}$, applying the unitary $U$ to the state $\ket{0}^{\otimes a}\ket{\psi}$, measuring the first $a$-qubit register and post-selecting on the $\ket{0}^{\otimes a}$ outcome, we get a state $\propto \rho\ket{\psi}$ in the second register. Block-encodings are easy to work with, for example given a block-encoding of $\rho$ and $\sigma$ we can easily construct a block-encoding of $(\rho-\sigma)/2$, see for example in the work of \cite{chakraborty2018BlockMatrixPowers}.

\paragraph{Example application to $\ell^3$-testing.} The problem is to decide whether $\rho=\sigma$ or $\nrm{\rho-\sigma}_3\geq \eps$, with query complexity $\bigO{\eps^{-\frac{3}{2}}}$. The first idea is that if we can prepare a purification of $\rho$ and $\sigma$, then by flipping a fair coin and preparing $\rho$ or $\sigma$ based on the outcome, we can also prepare a purification of $(\rho+\sigma)/2$. The second idea is to combine the block-encodings of $\rho$ and $\sigma$ to apply the map $\frac{\rho-\sigma}{2}$ to the purification of $(\rho+\sigma)/2$, to get
$$ \bigket{\frac{\rho+\sigma}{2}}\mapsto \left(\frac{\rho-\sigma}{2}\otimes I\right) \bigket{\frac{\rho+\sigma}{2}}\ket{0}+\ldots \ket{1}.$$
Finally, apply amplitude estimation with setting $M=\Theta(\eps^{-\frac{3}{2}})$. This works 
since if $\nrm{\rho-\sigma}_3\geq \eps$, then the $\ket{0}$ ancilla state has 
probability $\tr{(\rho-\sigma)^2(\rho+\sigma)}/8\geq \tr{|\rho-\sigma|^3}/8\geq \eps^3/8$.

\paragraph{Working with singular values.}
The above is a promising approach because it directly makes the density operator in question operationally accessible. However, it turns out that using this simple block-encodings is often suboptimal for distribution testing, because a query in some sense gives access to the square-root of $\rho$, whereas this unitary has $\rho$ itself in the top-left corner. Since the problems often heavily depend on smaller eigenvalues of $\rho$, the square root of $\rho$ is more desirable since it has quadratically larger singular-/eigenvalues.

Therefore, we show how to efficiently construct a unitary matrix whose top-left corner contains a matrix with singular values $\sqrt{p_{1}},\ldots,\sqrt{p_{n}}$, given purified access to a classical distribution $p$. To be more precise, we define \emph{projected unitary encodings}, which represents a matrix $A$ in the form of $\Pi U \widetilde{\Pi}$, where $\Pi, \widetilde{\Pi}$ are orthogonal projectors and $U$ is a unitary matrix. One can think about $U$ in a projected unitary encoding as a probabilistic implementation of the map $A\colon \img{\widetilde{\Pi}}\to \img{\Pi}$. Take for example $U:=\left(U_p\otimes I\right)$, $\Pi:=\left(\sum_{i=1}^{n}I\otimes \ketbra{i}{i}\otimes\ketbra{i}{i}\right)$, and $\widetilde{\Pi}:=\left(\ketbra{0}{0}\otimes\ketbra{0}{0}\otimes I\right)$. As we show in~\apx{ProjectedEncodings} these operators form a projected unitary encoding of
\begin{equation}\label{eq:probDistSingencoding}
	A=\Pi U \widetilde{\Pi}=\sum_{i=1}^{n}\sqrt{p_i}\ketbra{\phi_i}{0}\otimes\ketbra{i}{0}\otimes\ketbra{i}{i}.
\end{equation}

We can use a similar trick for a general density operator $\rho$ too. However, there is a major difficulty which arises from the fact that we do not a prioiri know the diagonalizig basis of $\rho$. Therefore we use slightly different operators. 
Let $W$ be a unitary,\footnote{This unitary is easy to implement, e.g., by using a few Hadamard and CNOT gates.} mapping $\ket{0}\ket{0}\mapsto\sum_{j=1}^{n}\frac{\ket{j}\ket{j}}{\sqrt{n}}$. Let $U':=\!\left(I\otimes U^\dagger_\rho\right)\!\left(W^\dagger \otimes I\right)$, $\Pi':=\left(I\otimes\ketbra{0}{0}\otimes\ketbra{0}{0}\right)$ and $\widetilde{\Pi}$ as above.
As we show in~\apx{ProjectedEncodings} these operators form a projected unitary encoding of
\begin{equation}\label{eq:densityOpSingencoding}
A'=\Pi U'\widetilde{\Pi}
=\sum_{i=1}^{n}\sqrt{\frac{p_i}{n}}\ketbra{\phi'_i}{0}\otimes\ketbra{0}{0}\otimes\ketbra{0}{\psi_i}.
\end{equation}

As we can see, the case of general density operators is less efficient, it only gives operational access to the ``square root'' of $\rho/n$. If the $1/\sqrt{n}$ factor could be directly improved, that would speed up our von Neumann entropy estimation algorithm \thm{quantum-entropy-purified}, which seems unlikely, cf.~\tab{entropy}.

\paragraph{General recipe.}
Our recipe to distributional property testing can be summarized as follows.
\begin{enumerate}[leftmargin=*,itemsep=0.85mm,label={\arabic*.)}]
	\item Construct a unitary matrix / quantum circuit operationally representing the distribution.\kern-10mm
	\item Transform the singular values of the corresponding matrix according to a desired function.
	\item Apply the resulting map to the purification of the distribution, or another suitable state.
	\item Estimate the amplitude of the flagged output state and conclude.
\end{enumerate}
The above general scheme describes our approach to the problems we discuss in this paper. Sometimes it is useful to divide the probabilities / singular values into bins, and fine-tune the algorithm by using the approximate knowledge of the size of the singular values. This divide-and-conquer strategy is at the core of our improved robust $\ell^2$-closeness tester of \thm{classical-l2-purified}.

\subsection{Related works on distributional property testing}\label{sec:related-works}
\begin{paragraph}{Classical algorithms.}
Many distributional property testing problems fall into the category of \emph{closeness testing}, where we are given the ability to take independent samples from two unknown distributions $p$ and $q$ with cardinality $n$, and the goal is to determine whether they are the same versus significantly different. For \emph{$\ell^{1}$-closeness testing}, which is about testing whether $p=q$ or $\|p-q\|_{1}\geq\epsilon$, \cite{batu2013TestClosenessDiscDistr} first gave a sublinear algorithm using $\tilde{O}(n^{2/3}/\epsilon^{8/3})$ samples to $p$ and $q$. The follow-up work by \cite{chan2014OptimalAlgTestCloseDistr} determined the optimal sample complexity as $\Theta(\max\{\frac{n^{2/3}}{\epsilon^{4/3}},\frac{n^{1/2}}{\epsilon^{2}}\})$; the same paper also gave a tight bound $\Theta(\frac{1}{\epsilon^{2}})$ for \emph{$\ell^{2}$-closeness testing}.

Besides closeness testing, a similar problem is \emph{identity} testing where one of the distributions, say $q$, is known and we are given independent samples from the other distribution $p$. For \emph{$\ell^{1}$ identity testing}, it is known that the sample complexity can be smaller than that of $\ell^{1}$-closeness testing, which was proved by \cite{batu2001TestRndVarIndepIdentity} to be $\tilde{O}(\sqrt{n}/\epsilon^{4})$ and then \cite{paninski2008CoincidenceBasedTestUnif} gave the tight bound $\Theta(\sqrt{n}/\epsilon^{2})$. More recently, \cite{diakonikolas2016NewTestPropDistr} proposed a modular reduction-based approach for distributional property testing problems, which recovered all closeness and identity testing results above. Furthermore, they also studied \emph{independence testing}, i.e., whether a distribution on $\range{n}\times\range{m}$ ($n\geq m$) is a product distribution or at least $\epsilon$-far in $\ell^{1}$-distance from any product distribution, and determined the optimal bound $\Theta(\max\{\frac{n^{2/3}m^{1/3}}{\epsilon^{4/3}},\frac{(nm)^{1/2}}{\epsilon^{2}}\})$.

Apart from the relationship between distributions, properties of a single distribution also have been extensively studied. One of the most important properties is \emph{Shannon entropy} (\cite{shannon1948MatheTheoryOfComm}) because it measures for example compressibility. The sample complexity of estimating $H(p)$ within additive error $\eps$ has been intensively studied (\cite{batu2005ComplOfApxEntropy,paninski2003EstiEntrMutInf,paninski2004EstEntrSubLinSamp}); in particular, \cite{valiant2011EstimatingTheUnseen,valiant2011PowerOfLinearEstimators} gave an explicit algorithm for entropy estimation using $\Theta(\frac{n}{\epsilon\log n})$ samples when $\epsilon=\Omega(n^{-0.03})$ and $\epsilon=O(1)$; for the general case \cite{jiao2015MinimaxEstFunDiscDistr} and \cite{wu2016MinimaxRatesEntrEstBestPolyApx} gave the optimal estimator with $\Theta\left(\frac{n}{\epsilon\log n}+\frac{(\log n)^{2}}{\epsilon^{2}}\right)$ samples.
\end{paragraph}

\begin{paragraph}{Quantum algorithms.}
The first paper on distributional property testing by quantum algorithms was by \cite{bravyi2011QAlgTestingPropOfDistr}, which considered classical distributions with discrete quantum query-access (see \defn{discrete-query}); it gives a quantum query complexity upper bound $O(\sqrt{n}/\epsilon^{6})$ for $\ell^{1}$-closeness testing and $O(n^{1/3}/\epsilon^{4/3})$ for identity testing to the uniform distribution on $\range{n}$. Subsequently, \cite{chakraborty2010NewResQuantPropTest} gave an algorithm for identity testing (to an arbitrary known distribution) with $\bOt{n^{1/3}/\epsilon^{5}}$ queries, and \cite{montanaro2015QMonteCarlo} improved the $\epsilon$-dependence of $\ell^{1}$-closeness testing to $\bOt{\sqrt{n}/\epsilon^{2.5}}$. More recently, \cite{li2017QQueryEntropyComp} studied entropy estimation under this model, and gave a quantum algorithm for Shannon entropy estimation with $\bOt{\sqrt{n}/\epsilon^{2}}$ queries and also sublinear quantum algorithms for estimating R{\'e}nyi entropies (\cite{renyi1961Entropy}).

Another type of quantum property testing results (\cite{odonnell2015QuantumSpectrumTesting,odonnell2016EfficientQuantumTomography,odonnell2017EfficientQuantumTomographyII,badescu2017QStateCertification,acharya2017MeasuringQuantEntropy}) concern \emph{density matrices}, where the $\ell^{1}$-distance becomes the trace distance and the Shannon entropy becomes the von Neumann entropy. To be more specific, for $n$-dimensional density matrices, the number of samples needed for $\ell^{1}$ and $\ell^{2}$-closeness testing are $\Theta(n/\epsilon^{2})$ and $\Theta(1/\epsilon^{2})$ (\cite{badescu2017QStateCertification}), respectively. In addition \cite{acharya2017MeasuringQuantEntropy} gave upper and lower bounds $\bigO{n^{2}/\epsilon^{2}}, \Omega\left(n^{2}/\epsilon\right)$ for estimating the von Neumann entropy of an $n$-dimensional density matrix with accuracy $\epsilon$.
\end{paragraph}

\subsection{Organization of the paper}
The rest of the paper is organized as follows. In \sec{prelim} we introduce two important quantum algorithmic techniques, amplitude estimation and singular value transformation. We give entropy estimators of classical and quantum distributions in \sec{entropy}. In \sec{testl2equality} we give an (essentially) optimal quantum algorithm for robustly testing $\ell^{2}$-closeness of classical distributions, and another efficient robust $\ell^2$-closeness tester for quantum distributions. Proof details of projected encodings, polynomial approximations for singular value transformation, and corollaries about $\ell^1$-closeness and independence testing are deferred to \apx{ProjectedEncodings}, \ref{apx:SVT-appendix}, and \ref{apx:proof-appendix} respectively.


\section{Preliminaries}\label{sec:prelim}

\subsection{Amplitude estimation}
Classically, given i.i.d.~samples of a Bernoulli random variable $X$ with $\Ex[X]=p$, it takes $\Theta(1/\epsilon^{2})$ samples to estimate $p$ within $\epsilon$ with high success probability. Quantumly, if we are given a unitary $U$ such that
\begin{align}\label{eq:AmpEst-def}
U|0\>|0\>=\sqrt{p}|0\>|\phi\>+|0^{\perp}\>,\quad\text{ where } \nrm{\ket{\phi}}=1 \text{ and }(\<0|\otimes I)|0^{\perp}\>=0,
\end{align}
then if measure the output state, we get $0$ in the first register with probability $p$.
Given access to $U$ we can estimate the value of $p$ quadratically more efficiently than what is possible by sampling:
\begin{theorem}{\!\bf\cite[Theorem 12]{brassard2002AmpAndEst}}\label{thm:AmpEst}
Given $U$ satisfying \eq{AmpEst-def}, the amplitude estimation algorithm outputs $\tilde{p}$ such that $\tilde{p}\in[0,1]$ and
\begin{align}\label{eq:AmpEst-1}
|\tilde{p}-p|\leq \frac{2\pi\sqrt{p(1-p)}}{M}+\frac{\pi^{2}}{M^{2}}
\end{align}
with success probability at least $8/\pi^{2}$, using $M$ calls to $U$ and $U^{\dagger}$.
\end{theorem}

In particular, if we take $M=\left\lceil 2\pi\left(\frac{2\sqrt{p}}{\epsilon}+ \frac{1}{\sqrt{\epsilon}}\right) \right\rceil=\Theta\left(\frac{\sqrt{p}}{\epsilon}+ \frac{1}{\sqrt{\epsilon}}\right)$ in \eq{AmpEst-1}, we have
\begin{align*}
|\tilde{p}-p|\leq\frac{2\pi\sqrt{p(1-p)}}{2\pi}\epsilon+\frac{\pi^{2}}{4\pi^{2}}\epsilon^{2}\leq\frac{\epsilon}{2}+\frac{\epsilon}{4}\leq\epsilon.
\end{align*}
Therefore, using only $\Theta(1/\epsilon)$ implementations of $U$ and $U^{\dagger}$, we could get an $\epsilon$-additive approximation of $p$ with success probability at least $8/\pi^{2}$, which is a quadratic speed-up compared to the classical sample complexity $\Theta(1/\epsilon^{2})$. The success probability can be boosted to $1-\nu$ by executing the algorithm for $\Theta(\log 1/\nu)$ times and taking the median of the estimates.

\subsection{Quantum singular value transformation}
Singular value decomposition (SVD) is one of the most important tools in linear algebra, generalizing eigen-decomposition of Hermitian matrices. Recently, \cite{gilyen2018QSingValTransf} proposed \emph{quantum singular value transformation} which turns our to be very useful for property testing. Mathematically, it is defined as follows:
\begin{definition}[Singular value transformation]\label{def:PolySVTrans}
		Let $f:\mathbb{R}\rightarrow\mathbb{C}$ be an even or odd function.
		Let $A\in\C^{\tilde{d}\times d}$ have the following singular value decomposition
		\cvskip{-3mm}
		\begin{equation*}
		A=\sum_{i=1}^{d_{\min}}\varsigma_i\ketbra{\tilde{\psi}_i}{\psi_i},
		\end{equation*}
		where $d_{\min}:=\min(d,\tilde{d})$. For the function $f$ we define the \emph{singular value transformation} on $A$ as
		\cvskip{-4mm}
		\begin{equation*}
			f^{(SV)}(A):=\left\{\begin{array}{rcl} \sum_{i=1}^{d_{\min}}f(\varsigma_i)\ketbra{\tilde{\psi}_i}{\psi_i}& &\text{if }f\text{ is odd, and}\\[\medskipamount]
			\sum_{i=1}^{d}f(\varsigma_i)\ketbra{\psi_i}{\psi_i}& & \text{if }f\text{ is even, where for } i\in[d]\setminus[d_{\min}]\text{ we define }\varsigma_i:=0.\end{array}\right.
		\end{equation*}		
\end{definition}

Quantum singular value transformation by real polynomials can be efficiently implemented on a quantum computer as follows:
\begin{theorem}{\!\bf\cite[Corollary 18]{gilyen2018QSingValTransf}}\label{thm:matchingParity}
Let $\H_U$ be a finite-dimensional Hilbert space and let $U,\Pi, \widetilde{\Pi}\in\eend{\H_U}$ be linear operators on $\H_U$ such that $U$ is a unitary, and $\Pi, \widetilde{\Pi}$ are orthogonal projectors. Suppose that $P=\sum_{k=0}^{n}a_k x^k\in\R[x]$ is a degree-$n$ polynomial such that
	\begin{itemize}
		\item $a_k\neq 0$ only if $\,k \equiv n \mod 2$, and
		\item for all $x\in[-1,1]\colon$ $| P(x)|\leq 1$.
	\end{itemize}
	Then there exist $\Phi\in\R^n$, such that
	\begin{equation*}
	P^{(SV)}\!\left(\widetilde{\Pi}U\Pi\right)\!=\!\left\{\begin{array}{rcl} \left(\bra+\otimes\widetilde{\Pi}\right)\Big(\ketbra00\!\otimes\! U_{\Phi}+\ketbra11\!\otimes\! U_{\!-\Phi}\Big) \left(\ket+\otimes\overset{\phantom{.}}{\Pi}\right)\!\!\!\!\!& &\text{if }n\text{ is odd, and}\\[\medskipamount]
	\left(\bra+\otimes\underset{\phantom{.}}{\Pi}\right)\Big(\ketbra00\!\otimes\! U_{\Phi}+\ketbra11\!\otimes\! U_{\!-\Phi}\Big) \left(\ket+\otimes\underset{\phantom{.}}{\Pi}\right)\!\!\!\!\!& & \text{if }n\text{ is even,}\end{array}\right.
	\end{equation*}
where $U_{\Phi}=e^{i\phi_{1} (2\widetilde{\Pi}-I)}U
			\prod_{j=1}^{(n-1)/2}\left(e^{i\phi_{2j} (2\Pi-I)}U^\dagger e^{i\phi_{2j+1} (2\widetilde{\Pi}-I)}U\right)$.\footnote{This is the mathematical form for odd $n$; even $n$ is defined similarly.}
\end{theorem}

Thus for an even or odd polynomial $P$ of degree $n$, we can apply singular value transformation of the matrix $\widetilde{\Pi} U \Pi$ with $n$ uses of $U$, $U^\dagger$ and the same number of controlled reflections $I\!-\!2\Pi, I\!-\!2\widetilde{\Pi}$.

To apply singular value transformation corresponding to our problems, we need low-degree polynomial approximations to the following functions, which we construct in \apx{SVT-appendix}.

\begin{restatable}{lemma}{polyApx}\kern-1mm\emph{\textbf{(Polynomial approximations)}}\label{cor:polyApx}
	Let $\beta\in(0,1]$, $\eta\in(0,\frac{1}{2}]$ and $t\geq 1$. There exists polynomials $\tilde{P},\tilde{Q},\tilde{S}$ such that
	\begin{itemize}
		\item $\forall x\in [\frac{1}{t},1]\colon |\tilde{P}(x)-\frac{1}{2t x}|\leq\eta$, and $\,\forall x\in[-1,1]\colon -1\leq \tilde{P}(x)=\tilde{P}(-x)\leq 1$,
		\item $\forall x\in [-\frac{1-\beta}{t},\frac{1-\beta}{t}]\colon |\tilde{Q}(x)-tx|\leq\eta \cdot (tx)$, and $\,\forall x\in[-1,1]\colon \tilde{Q}(x)=\!-\tilde{Q}(-x)\leq 1$,
		\item $\forall x\in [\beta,1]\colon |\tilde{S}(x)-\frac{\ln(1/x)}{2\ln(2/\beta)}|\leq\eta$,	and $\,\forall x\in[-1,1]\colon -1\leq\tilde{S}(x)=\tilde{S}(-x)\leq 1$,		
	\end{itemize}
	moreover $\deg(\tilde{P})=\bigO{t\log\left(\frac{1}{\eta}\right)}$, $\deg(\tilde{Q})=\bigO{\frac{t}{\beta}\log\left(\frac{1}{\eta}\right)}$, and $\deg(\tilde{S})=\bigO{\frac{1}{\beta}\log\left(\frac{1}{\eta}\right)}$.
\end{restatable}


\section{Entropy estimation}\label{sec:entropy}
\subsection{Classical distributions with purified quantum query-access}\label{sec:classical-entropy-purified}
Recall that we introduced purified quantum query-access in \defn{purified-quantum-query}. In particular, for a classical distribution $p$ on $\range{n}$, we are given a unitary $U_{p}$ acting on $\C^{n\times n}$ such that
\begin{align}\label{eq:classical-entropy-purified}
U_{p}|0\>_A|0\>_B=\ket{\psi_p}=\sum_{i=1}^{n}\sqrt{p_i}|\phi_i\>_A|i\>_B.
\end{align}
We use $U_{p}$ and $U_{p}^{\dagger}$ to estimate the Shannon entropy $H(p)$:
\begin{theorem}\label{thm:classical-entropy-purified}
For any $0<\epsilon<1$, we can estimate $H(p)$ with accuracy $\epsilon$ with success probability at least $2/3$ using $\bigO{\frac{\sqrt{n}}{\epsilon^{1.5}}\log^{1.5}\!\left(\frac{n}{\eps}\right)\log\left(\frac{\log n}{\eps}\right)}$ calls to $U_{p}$ and $U_{p}^{\dagger}$.
\end{theorem}

\begin{proof}
The general idea is to first construct a unitary matrix with singular values $\sqrt{p_{1}},\ldots,\sqrt{p_{n}}$. We use the construction of Eq.~\eqref{eq:probDistSingencoding} and apply singular value transformation (\thm{matchingParity}) by a polynomial $\tilde{S}$ constructed in \cor{polyApx}, setting $\eta=\frac{\epsilon}{24\ln(2/\beta)}$ and $\beta=\sqrt{\Delta}$ for $\Delta=\frac{\epsilon}{4n\ln (n/\epsilon)}$. Notice that this $\Delta$ satisfies
\begin{align}\label{eq:classical-entropy-Delta}
\Delta\Big(\ln\Big(\frac{1}{\Delta}\Big)+1\Big)&=\frac{\epsilon}{4n\ln (n/\epsilon)}\cdot\ln\frac{4en\ln(n/\epsilon)}{\epsilon}\leq\frac{\epsilon}{4n\ln (n/\epsilon)}\cdot\ln\frac{n^{2}}{\epsilon^{2}}=\frac{\epsilon}{2n},
\end{align}
provided that $\frac{n}{\eps}\geq 42$. Note that the polynomial $\tilde{S}$ satisfies both conditions in \thm{matchingParity}. Applying the singular value transformed version of the operator \eqref{eq:probDistSingencoding} to the state $\ket{\psi_p}$ results in 
\begin{align}\label{eq:classical-entropy-purified-transformed}
|\widetilde{\Psi_{p}}\>=\sum_{i=1}^{n}\sqrt{p_i}\tilde{S}(\sqrt{p_i})|\phi_i\>_A|i\>_B|0\> + \ldots|1\>.
\end{align}
Preparing $|\widetilde{\Psi_{p}}\>$ costs $\deg \tilde{S}=\bigO{\frac{1}{\beta}\log\left(\frac{1}{\eta}\right)}=\bigO{\sqrt{\frac{n}{\epsilon}\log\left(\frac{n}{\eps}\right)}\log\left(\frac{\log n}{\eps}\right)}$ uses of $U_{p}$ and $U_{p}^{\dagger}$ and the same number of controlled reflections through $\Pi, \widetilde{\Pi}$.
Furthermore, Eq. \eq{truncApx-symmetry-error} implies that for all $i$ such that $p_{i}\geq\Delta$,
\begin{align}\label{eq:classical-entropy-error-largeprob}
\Big|\frac{p_{i}\ln(1/p_{i})}{4\ln(2/\beta)}-p_{i}\tilde{S}(\sqrt{p_{i}})\Big|=p_{i}\cdot\Big|\frac{\ln(1/\sqrt{p_{i}})}{2\ln(2/\beta)}-\tilde{S}(\sqrt{p_{i}})\Big|\leq\eta p_{i}.
\end{align}
For all $i$ such that $p_{i}<\Delta$, we have
\begin{align}\label{eq:classical-entropy-error-smallprob}
\Big|\frac{p_{i}\ln(1/p_{i})}{4\ln(2/\beta)}-p_{i}\tilde{S}(\sqrt{p_{i}})\Big|\leq\frac{p_{i}\ln(1/p_{i})+p_{i}}{4\ln(2/\beta)}\leq\frac{\Delta(\ln(\frac{1}{\Delta})+1)}{4\ln(2/\beta)}\leq\frac{\epsilon}{8n\ln(2/\beta)},
\end{align}
where the first inequality comes from the fact that $|\tilde{S}(x)|\leq 1$ for all $x\in[-1,1]$, the second inequality comes from the monotonicity of $x(\ln(1/x)+1)$ on $(0,\frac{1}{\Delta}]$, and the third inequality comes from \eq{classical-entropy-Delta}. As a result of \eq{classical-entropy-purified}, \eq{classical-entropy-error-largeprob}, and \eq{classical-entropy-error-smallprob}, we have
\begin{align*}
\left|\big(\<\psi_{p}|\otimes\<0|\big)|\widetilde{\Psi_{p}}\>-\frac{H(p)}{4\ln(2/\beta)}\right|&=\left|p_{i}\tilde{S}(\sqrt{p_{i}})-\sum_{i=1}^{n}\frac{p_{i}\log(1/p_{i})}{4\ln(2/\beta)}\right| \\
&\leq\sum_{i\colon p_{i}<\Delta}\frac{\epsilon}{8n\ln(2/\beta)}+\sum_{i\colon p_{i}\geq\Delta}\eta p_{i} \\
&\leq\frac{\epsilon}{8\ln(2/\beta)}+\frac{\epsilon}{24\ln(2/\beta)}=\frac{\epsilon}{6\ln(2/\beta)}.
\end{align*}
Therefore, $|4\ln(2/\beta)(\<\psi_{p}|\otimes\<0|)|\widetilde{\Psi_{p}}\>-H(p)|\leq 2\epsilon/3$. By \thm{AmpEst}, we can use $\Theta(\ln(1/\beta)/\epsilon)$ applications of the unitaries (and their inverses) that implement $|\psi_{p}\>$ and $|\widetilde{\Psi_{p}}\>$ to estimate $(\<\psi_{p}|\otimes\<0|)|\widetilde{\Psi_{p}}\>$ within additive error $\frac{\epsilon}{12\ln(2/\beta)}$. In total, this estimates $H(p)$ within additive error $\frac{\epsilon}{12\ln(2/\beta)}\cdot 4\ln(2/\beta)+\frac{2\epsilon}{3}=\epsilon$ with success probability at least $8/\pi^{2}$. The total complexity of the algorithm is
\begin{equation*}
\bigO{\frac{\ln(1/\beta)}{\epsilon}}\cdot\bigO{\sqrt{\frac{n}{\epsilon}\log\left(\frac{n}{\eps}\right)}\log\left(\frac{\log n}{\eps}\right)}=\bigO{\frac{\sqrt{n}}{\epsilon^{1.5}}\log^{1.5}\!\left(\frac{n}{\eps}\right)\log\left(\frac{\log n}{\eps}\right)}. \qedhere
\end{equation*}
\cvskip{-5mm}
\end{proof}

\subsection{Density matrices with purified quantum query-access}\label{sec:quantum-entropy-purified}
For a density matrix $\rho$, we also assume the purified quantum query-access in \defn{purified-quantum-query}, i.e., a unitary oracle $U_{\rho}$ acting as $U_{\rho}\ket{0}_A\ket{0}_B=\ket{\rho}=\sum_{i=1}^{n}\sqrt{p_i}\ket{\phi_i}_A\ket{\psi_i}_B$. We use $U_{\rho}$ and $U_{\rho}^{\dagger}$ to estimate the von-Neumann entropy $H(\rho)=-\Tr[\rho\log\rho]$:
\begin{theorem}\label{thm:quantum-entropy-purified}
For any $0<\epsilon<1$, we can estimate $H(p)$ with accuracy $\epsilon$ with success probability at least $2/3$ using $\bOt{\frac{n}{\eps^{1.5}}}$ calls to $U_{\rho}$ and $U_{\rho}^{\dagger}$.
\end{theorem}

\begin{proof}
We use the construction of Eq.~\eqref{eq:densityOpSingencoding}. The proof is essentially the same as that of \thm{classical-entropy-purified} proceeding by constructing singular value transformation via \thm{matchingParity}, with the only difference that all probabilities are rescaled by a factor of $1/\sqrt{n}$ in \eq{densityOpSingencoding}; as a result, the number of calls to $U_{\rho}$ and $U_{\rho}^{\dagger}$ is blown up to $\bOt{\sqrt{n}\cdot\frac{\sqrt{n}}{\eps^{1.5}}}=\bOt{\frac{n}{\eps^{1.5}}}$.
\end{proof}


\section{Robust testers for \texorpdfstring{$\ell^{2}$}{l2}-closeness with purified query-access}\label{sec:testl2equality}
First we give an $\ell^{2}$-closeness tester for unknown classical distributions $p,q$.
\begin{theorem}\label{thm:classical-l2-purified}
Given purified quantum query-access for classical distributions $p,q$ as in \defn{purified-quantum-query}, for any $\nu,\epsilon\in(0,1)$ the quantum query complexity of distinguishing the cases $\nrm{p-q}_{2}\geq\eps$ and $\nrm{p-q}_{2}\leq(1-\nu)\eps$ with success probability at least $2/3$ is 
$\bigO{\frac{1}{\nu\eps}\log^3\left(\frac{1}{\nu\eps}\right)\log\log\left(\frac{1}{\nu\eps}\right)}$.
\end{theorem}

\begin{proof}
The main idea is to first bin the $x$ elements based on the approximate value of $p(x)+q(x)$, then apply fine-tuned algorithms exploiting the knowledge of the approximate value of $p(x)+q(x)$.

Using amplitude estimation for any $k\in\N$  we can construct an algorithm $\A_k$ that for any input $x$ with $p(x)+q(x)\geq 2^{-k}$ outputs ``greater'' with probability at least $2/3$, and for any $x$ with $p(x)+q(x)\leq 2^{-k-1}$ outputs ``smaller'' and uses $\bigO{2^{\frac{k}{2}}}$ queries to $U_{p}$ and $U_{q}$.  Using $\bigO{\log(\frac{1}{\nu\eps}))}$ repetitions we can boost the success probability to $1-\bigO{\poly\left(\nu\eps\right)}$. Since our algorithm only needs to succeed with constant probability, and will use these subroutines at most $\frac{1}{\poly\left(\nu\eps\right)}$ times, we can ignore the small failure probability. Therefore in the rest of the proof we assume without loss of generality, that $\A_k$ that solves perfectly the above question with (query) complexity  $\bigO{2^{\frac{k}{2}}\log(\frac{1}{\nu\eps}))}$.

\begin{algorithm}[H]
	\caption{Estimating $\log_2(p(x)+q(x))$}\label{alg:Magnitude}
	\begin{algorithmic}[1]
		\STATEx \textbf{input} $x\in[n]$, $\theta\in (0,1)$
		\STATE \textbf{for} $k\in K:=\left\{-1,0,1,2,\ldots, \left\lceil\log_2\left(\frac{1}{\theta}\right)\right\rceil\right\}$ \textbf{do}
		\STATE \qquad Run algorithm $\A_k$ on $\ket{x}$ \textbf{if} output is ``greater'' \textbf{then return} $k$
		\STATE \textbf{return} ``less than $\theta$''		
	\end{algorithmic}
\end{algorithm}
\noindent
For any $x$ with $p(x)+q(x)\geq \theta$, \alg{Magnitude} outputs a $k$ such that $p(x)+q(x)\in (2^{-k-1},2^{-k+1})$. 
However, note that this labeling is probabilistic; let us denote by $s_k(x)$ the probability that $x$ is labeled by $k$. Observe that $s_k(x)=0$ unless $k\in \left\{\left\lfloor\log_2\left(\frac{1}{p(x)+q(x)}\right)\right\rfloor,\left\lceil\log_2\left(\frac{1}{p(x)+q(x)}\right)\right\rceil\right\}$. Now let us express $\nrm{p-q}^2_2$ in terms of this ``soft-selection'' function $s(x)$.
\begin{align}
\nrm{p-q}^2_2&=\sum_{x}\left|p(x)-q(x)\right|^2\nonumber\\
&=\sum_{x}\sum_{k\in K}s_k(x)\left|p(x)-q(x)\right|^2 +\eta \tag*{$\eta\in [0,2\theta)$}\nonumber\\
&=\sum_{k\in K}2^{9-k}\sum_{x}s_k(x)\frac{p(x)+q(x)}{2}\frac{2^{-k-2}}{p(x)+q(x)}\left(\frac{p(x)-q(x)}{2^{-k+3}}\right)^{\!\!2} +\eta,
\end{align}
where the bound on $\eta$ follows from the observation that
$$\eta
\leq \!\! \sum_{x\colon p(x)+q(x)< \theta}\!\! \left|p(x)-q(x)\right|^2
\leq \!\! \sum_{x\colon p(x)+q(x)< \theta}\!\! \left(p(x)+q(x)\right)^2
< \theta \sum_{x\colon p(x)+q(x)< \theta}p(x)+q(x)< 2\theta.$$

If for all $k\in K$ we have a $2^{k-9}\frac{\theta}{|K|}$-precise estimate of 
\begin{equation}
	\sum_{x}s_k(x)\frac{p(x)+q(x)}{2}\frac{2^{-k-2}}{p(x)+q(x)}\left(\frac{p(x)-q(x)}{2^{-k+3}}\right)^{\!\!2},\label{eq:l2prob}
\end{equation}
then we get a $3\theta$-precise estimate of $\nrm{p-q}^2_2$. In particular setting $\theta:= \nu\eps^2/6$, this solves the robust testing problem, since if $\nrm{p-q}\geq \eps$ then $\nrm{p-q}^2\geq \eps^2$, on the other hand if $\nrm{p-q}\leq (1-\nu)\eps$ then $\nrm{p-q}^2\leq (1-\nu)^2\eps^2\leq (1-\nu)\eps^2= \eps^2 -\nu\eps^2$.

Now we describe how to construct a quantum algorithm that sets the first output qubit to $\ket{0}$ with probability \eqref{eq:l2prob}. Start with preparing a purification of the distribution of $\frac{p(x)+q(x)}{2}$, then set the label of $x$ to $k$ with probability $s_k(x)$ using \alg{Magnitude} terminating it after using $\A_k$. Then separately apply the maps $\sqrt{\frac{2^{-k-2}}{p(x)+q(x)}}$ and $\frac{p(x)-q(x)}{2^{-k-3}}$ to the state.

Note that we do not need to apply the above transformations exactly, it is enough if apply them with precision say $2^{k-11}\frac{\theta}{|K|}$. 
We analyze the complexity of (approximately) implementing the above sketched algorithm.
To implement the map $\sqrt{\frac{2^{-k-2}}{p(x)+q(x)}}$, we use the unitary of Eq.~\eqref{eq:probDistSingencoding}, and transform the singular values by the polynomial $\tilde{P}$ from \cor{polyApx} using \thm{matchingParity}. In order to implement the map $\frac{p(x)-q(x)}{2^{-k-2}}$, we again use the unitary of Eq.~\eqref{eq:probDistSingencoding}, but now separately for $p$ and $q$. We amplify both the singular values $\sqrt{p(x)}$ and $\sqrt{q(x)}$ by a factor $\sqrt{2^{k-2}}$ using the polynomial $\tilde{Q}$ from \cor{polyApx} in \thm{matchingParity}. Then we create a bolck-encoding\footnote{If we have a projected unitary encoding of $\Pi U \widetilde{\Pi}=A = \sum_{i}\varsigma_i\ketbra{\psi_i}{0,i}$ with $\widetilde{\Pi}=\ketbra{0}{0}\otimes I$ we can immediately turn it into a block-encoding of $A^\dagger A=\sum_{i}\varsigma^2_i\ketbra{i}{i}$, by e.g. applying \thm{matchingParity} with the polynomial $x^2$. } of both and $2^{k-2}p(x)$ and $2^{k-2}q(x)$ and then combine them to get a block-encoding of $\frac{p(x)-q(x)}{2^{-k-3}}$. In both cases the query complexity of $\bigO{\theta/|K|}$-precisely implementing the transformations is $\bigO{2^{k/2}\log\left(|K|/\theta\right)}=\bigO{2^{k/2}\log\left(1/\theta\right)}$.
Since computing the label $k$ also costs $\bigO{2^{k/2}\log\left(1/(\nu\eps)\right)}$, this is the overall complexity so far.
Finally we estimate the probability of the first qubit being set to $\ket{0}$ with setting $M=\bigO{|K|2^{-k/2}/(\nu\eps)}$ in \thm{AmpEst}, and boost the success probability to $1-\bigO{1/|K|}$ with $\bigO{\log(|K|)}$ repetitions. Thus for any $k\in K$ the overall complexity of estimating Eq.~\eqref{eq:l2prob} with sufficient precision has (query) complexity $\bigO{\frac{|K|}{\nu\eps}\log\left(\frac{1}{\nu\eps}\right)\log(|K|)}=\bigO{\frac{1}{\nu\eps}\log^2\left(\frac{1}{\nu\eps}\right)\log\log\left(\frac{1}{\nu\eps}\right)}$. Therefore estimating $\nrm{p-q}^2_2$ to precision $\nu\eps^2/6$ with high probability has (query) complexity
\begin{equation*}
\bigO{\frac{1}{\nu\eps}\log^3\left(\frac{1}{\nu\eps}\right)\log\log\left(\frac{1}{\nu\eps}\right)}. \qedhere
\end{equation*}
\cvskip{-5mm}
\end{proof}

It is easy to see an $\Omega\left(\frac{1}{\eps}\right)$ lower bound on the above problem even in the strongest quantum pure state input model~\defn{pureSTateQuery}. Indeed, consider the case $n=2, q=(\frac{1}{2},\frac{1}{2})$ (the uniform distribution on $\{1,2\}$) and we want to test whether $p=q$ or $\|p-q\|_{2}\geq\eps$. This is equivalent to test whether $p_{1}=\frac{1}{2}$ or $|p_{1}-\frac{1}{2}|\geq\frac{\eps}{\sqrt{2}}$; due to the optimality of amplitude estimation in \thm{AmpEst}, this task requires $\Omega(\frac{1}{\eps})$ quantum queries to the unitary $U$ preparing the state $\sqrt{p_1}\ket{1}+\sqrt{p_2}\ket{2}$.

Now we prove the following result on (robust) $\ell^2$-closeness testing for quantum distributions:
\begin{theorem}\label{thm:quantum-l2-purified}
	Given $\eps,\nu\in(0,1)$ and two density operators $\rho,\sigma\in\C^{n\times n}$ with purified quantum query-access to $U_{\rho}$ and $U_{\sigma}$ as in \defn{purified-quantum-query}, it takes $\bigO{\kern-0.3mm\min\left(\frac{\sqrt{n}}{\epsilon},\frac{1}{\epsilon^2}\right)\frac{1}{\nu}}$ queries to $U_{\rho},U_{\rho}^{\dagger},U_{\sigma},U_{\sigma}^{\dagger}$ to decide whether  $\nrm{\rho\!-\!\sigma}_{2}\geq\epsilon$ or $\nrm{\rho\!-\!\sigma}_{2}\leq(1-\nu)\epsilon$, with success probability at least $2/3$.
\end{theorem}

\begin{proof}
	We can combine the block-encodings of $\rho$ and $\sigma$ to apply the map $\frac{\rho-\sigma}{2}$ to the maximally entangled state $\sum_{j=1}^{n}\frac{\ket{j}\ket{j}}{\sqrt{n}}$, which gives
	\cvskip{-3mm}
	\begin{align*}
	\sum_{j=1}^n \frac{\ket{j}\ket{j}}{\sqrt{n}}\rightarrow \left(\frac{\rho-\sigma}{2}\otimes I\right)\sum_{j=1}^n \frac{\ket{j}\ket{j}}{\sqrt{n}}\ket{0}+\ldots \ket{1}.
	\end{align*}
	\cvskip{-1.5mm}\noindent
	The probability of measuring the $|0\>$ ancilla state is
	\cvskip{-3mm}
	\begin{align*}
	\sum_{i,j=1}^{n}\frac{\bra{i}\bra{i}}{\sqrt{n}}\left(\frac{(\rho-\sigma)^{2}}{4}\otimes I\right)\frac{\ket{j}\ket{j}}{\sqrt{n}}=\frac{1}{4n}\sum_{i=1}^{n}\bra{i}(\rho-\sigma)^{2}\ket{i}=\frac{1}{4n}\Tr[(\rho-\sigma)^{2}].
	\end{align*}
	\cvskip{-1.5mm}\noindent
	Thus it suffices to apply amplitude estimation with $M=\Theta\left(\frac{\sqrt{n}}{\nu\epsilon}\right)$ calls to $U_{\rho},U_{\rho}^{\dagger},U_{\sigma},U_{\sigma}^{\dagger}$.
	
	On the other hand, we can estimate $\nrm{\rho-\sigma}^2_2$ by observing that $\nrm{\rho-\sigma}^2_2=\tr{(\rho-\sigma)^2}=\tr{\rho^2}-2\tr{\rho\sigma}+\tr{\sigma^2}$. 
	Since the success probability of the SWAP test (\cite{buhrman2001QuantumFingerprinting}) on input states $\rho,\sigma$ is $\frac{1}{2}\left(1+\tr{\rho\sigma}\right)$, we can individually  estimate the latter quantities with precision $\bigO{\nu\eps^2}$ using amplitude estimation (\thm{AmpEst}) with $\bigO{\frac{1}{\nu\eps^2}}$ queries to $U_{\rho},U_{\rho}^{\dagger},U_{\sigma},U_{\sigma}^{\dagger}$. As a result, we could decide whether $\nrm{\rho\!-\!\sigma}_{2}\geq\epsilon$ or $\nrm{\rho\!-\!\sigma}_{2}\leq(1-\nu)\epsilon$ using $\bigO{\frac{1}{\nu\eps^2}}$ queries.
	
	The result of \thm{quantum-l2-purified} hence follows by taking the minimum of the two complexities.
\end{proof}
\cvspace{-2mm}
\section{Future work and open questions}
Our paper raises a couple of natural open questions for future work. For example:
\begin{itemize}[leftmargin=*]
	\item Can we prove quantum lower bounds that match our upper bounds? For instance, can we prove an $\Omega\left(\frac{n}{\eps}\right)$ lower bound on estimating the von Neumann entropy in the purified quantum query-access model for density operators? Is~there a lower bound technique which naturally fits our purified quantum query input model?
	
	\item For which other distributional property testing problems can we get speed-ups using the presented methodology? 

\end{itemize}

\ignore{
\section*{Acknowledgments}
A.G. thanks Ronald de Wolf, Ignacio Cirac and Yimin Ge for useful discussion.
}


\bibliography{Bibliography}

\newpage

\appendix


\section{Projected unitary encodings used for singular value transformation}\label{apx:ProjectedEncodings}
First we handle the case of classical distributions. 
Let $U_p$ be a purified quantum oracle of a classical distribution $p$ as in~\defn{purified-quantum-query}, and let $U:=\left(U_p\otimes I\right)$, also let $\Pi:=\left(\sum_{i=1}^{n}I\otimes \ketbra{i}{i}\otimes\ketbra{i}{i}\right)$, $\widetilde{\Pi}:=\left(\ketbra{0}{0}\otimes\ketbra{0}{0}\otimes I\right)$, then
\begin{align}
\Pi U\widetilde{\Pi}=\Pi\left(U_p\otimes I\right)\widetilde{\Pi}=&\Big(\sum_{i=1}^{n}I\otimes\ketbra{i}{i}\otimes \ketbra{i}{i}\Big)(U_{p}\otimes I)\big(\ketbra{0}{0}\otimes\ketbra{0}{0}\otimes I\big) \nonumber \\
=&\sum_{i=1}^{n}\Big((I\otimes \ketbra{i}{i})U_{p}(\ketbra{0}{0}\otimes\ketbra{0}{0})\Big)\otimes\ketbra{i}{i} \nonumber \\
=&\sum_{i=1}^{n}\Big((I\otimes \ketbra{i}{i})\sum_{j=1}^{n}\sqrt{p_{j}}\ket{\phi_{j}}\ket{j}\bra{0}\bra{0}\Big)\otimes\ketbra{i}{i} \nonumber \\
=&\sum_{i=1}^{n}\sqrt{p_i}\ketbra{\phi_i}{0}\otimes\ketbra{i}{0}\otimes\ketbra{i}{i}.  \nonumber
\end{align}

Now we turn to quantum distributions where we do not know the diagonalizing basis of the density operator $\rho$.
Let $U_\rho$ be a purified quantum oracle of a quantum distribution $\rho$ as in~\defn{purified-quantum-query}, and $W$ a unitary, mapping $\ket{0}\ket{0}\mapsto\sum_{j=1}^{n}\frac{\ket{j}\ket{j}}{\sqrt{n}}$. Let $U':=\!\left(I\otimes U^\dagger_\rho\right)\!\left(W^\dagger \otimes I\right)$ , $\Pi':=\left(I\otimes\ketbra{0}{0}\otimes\ketbra{0}{0}\right)$ and $\widetilde{\Pi}$ as above, then
\begin{align*}
\Pi' U'\widetilde{\Pi}=\Pi' \left(I\otimes U^\dagger_\rho\right)\!\left(W^\dagger \otimes I\right)\widetilde{\Pi}=&\left(I\otimes (\ketbra{0}{0}\otimes\ketbra{0}{0}U^\dagger_\rho)\right) \left(\left(\sum_{j=1}^{n}\frac{\ket{j}\ket{j}}{\sqrt{n}}\right)\bra{0}\bra{0}\otimes I\right)\\
=&\left(I\otimes \sum_{i=1}^{n}\sqrt{p_i}\ket{0}\ketbra{0}{\phi_i}\bra{\psi_i}\right) \!\left(\!\left(\sum_{j=1}^{n}\frac{\ket{\phi'_j}\ket{\phi_j}}{\sqrt{n}}\right)\!\bra{0}\bra{0}\otimes I\right)\\
=&\sum_{i=1}^{n}\sqrt{\frac{p_i}{n}}\ket{\phi'_i}\ket{0}\ket{0}\bra{0}\bra{0}\bra{\psi_i},
\end{align*}
where $\sum_{j=1}^{n}\frac{\ket{\phi'_j}\ket{\phi_j}}{\sqrt{n}}=\sum_{j=1}^{n}\frac{\ket{j}\ket{j}}{\sqrt{n}}$ is the Schmidt decomposition of the maximally entangled state under the basis $(|\phi_{1}\>,\ldots,|\phi_{n}\>)$.


\section{Polynomial approximations for singular value transformation}\label{apx:SVT-appendix}
We use the following result based on local Taylor series:
\begin{lemma}{\!\bf\cite[Corollary 66]{gilyen2018QSingValTransf}}\label{cor:boundedTaylorApx}
	Let $x_0\in[-1,1]$, $r\in(0,2]$, $\nu\in(0,r]$ and let  $f\colon [-x_0-r-\nu,x_0+r+\nu]\rightarrow \C$ and be such that $f(x_0+x)=\sum_{\ell=0}^{\infty} a_\ell x^\ell$ for all $x\in\![-r-\nu,r+\nu]$. Suppose $B>0$ is such that $\sum_{\ell=0}^{\infty}(r+\nu)^\ell|a_\ell|\leq B$. Let $\eps\in\!\left(0,\frac{1}{2B}\right]$, then there is an efficiently computable polynomial $P\in \C[x]$ of degree $\bigO{\frac{1}{\nu}\log\left(\frac{B}{\eps}\right)}$ such that\footnote{For a function $g\colon\R\rightarrow\C$, and an interval $[a,b]\subseteq\R$, we define $\|g\|_{[a,b]}:=\max_{x\in[a,b]} |g(x)|$.}
	\begin{align*}
	\nrm{ f(x)- P(x) }_{[x_0-r,x_0+r]}&\leq \eps\\
	\nrm{P(x)}_{[-1,1]}&\leq  \eps + \nrm{f(x)}_{[x_0-r-\nu/2,x_0+r+\nu/2]}\leq \eps + B\\
	\nrm{P(x)}_{[-1,1]\setminus [x_0-r-\nu/2,x_0+r+\nu/2]}&\leq \eps.
	\end{align*}
\end{lemma}

We can use the above result to construct the following useful polynomial approximations.

\polyApx*
\begin{proof}
For the construction of the $\tilde{P}$ and $\tilde{Q}$ polynomials see Corollary~67 and Theorem~30 of \cite{gilyen2018QSingValTransf}, respectively. It remains to construct the polynomial $\tilde{S}$ above.

Denote $f(x)=\frac{\ln(1/x)}{2\ln(2/\beta)}$; by taking $\epsilon=\eta/2$, $x_{0}=1$, $r=1-\beta$, $\nu=\frac{\beta}{2}$, and $B=\frac{1}{2}$ in \cor{boundedTaylorApx}, we have a polynomial $S\in\C[x]$ of degree $\bigO{\frac{1}{\nu}\log(\frac{B}{\epsilon})}=\bigO{\frac{1}{\beta}\log(\frac{1}{\eta})}$ such that
\begin{align}
			\nrm{ f(x)- S(x)}_{[\beta,2-\beta]}&\leq\eta/2 \label{eq:truncApx-entropy} \\
			\nrm{S(x)}_{[-1,1]}&\leq B+\eta/2\leq(1+\eta)/2 \label{eq:truncBndAll-entropy} \\
			\nrm{S(x)}_{[-1,\frac{\beta}{2}]}&\leq\eta/2. \label{eq:truncBnd-entropy}
\end{align}
Note that $B=\frac{1}{2}$ is valid because the local Taylor series of $f(x)$ at $x=1$ is $\frac{1}{2\ln(2/\beta)}\sum_{l=1}^{\infty}\frac{(-1)^{l}x^{l}}{l}$, and as a result we could take
\begin{align*}
B=\frac{1}{2\ln(2/\beta)}\sum_{l=1}^{\infty}\frac{(1-\beta/2)^{l}}{l}&=-\frac{1}{2\ln(2/\beta)}\sum_{l=1}^{\infty}\frac{(-1)^{l-1}}{l}(-1+\beta/2)^{l} \nonumber \\
&=-\frac{1}{2\ln(2/\beta)}\ln\frac{\beta}{2}=\frac{1}{2}.
\end{align*}
However, $S$ is not an even polynomial in general; we instead take $\tilde{S}(x)=S(x)+S(-x)$ for all $x\in [-1,1]$. Then by \eq{truncApx-entropy} and \eq{truncBnd-entropy} we have
\begin{align}\label{eq:truncApx-symmetry-error}
\nrm{ f(x)-\tilde{S}(x)}_{[\beta,1]}\leq \nrm{f(x)-\tilde{S}(x)}_{[\beta,1]}+\nrm{\tilde{S}(-x)}_{[\beta,1]}\leq\frac{\eta}{2}+\frac{\eta}{2}=\eta.
\end{align}
Furthermore, $\tilde{S}$ is an even polynomial such that $\deg(\tilde{S})=\bigO{\frac{1}{\beta}\log(\frac{1}{\eta})}$; hence \eq{truncBndAll-entropy} and \eq{truncBnd-entropy} imply
\begin{align*}
\nrm{\tilde{S}(x)}_{[-1,1]}=\nrm{\tilde{S}(x)}_{[0,1]}\leq\nrm{S(x)}_{[0,1]}+\nrm{S(x)}_{[-1,0]}\leq\frac{1+\eta}{2}+\frac{\eta}{2}\leq 1
\end{align*}
given $\eta\leq 1/2$. (Finally we can take the real part of $\tilde{S}(x)$ if it has some complex coefficients.)
\end{proof}


\section{Corollaries of our \texorpdfstring{$\ell^2$}{l2}-closeness testing results}\label{apx:proof-appendix}
\subsection{\texorpdfstring{$\ell^1$}{l1}-closeness testing with purified query-access}\label{sec:closeness-purified-quantum-query}
\begin{corollary}\label{cor:l1-purified}
	Given $\epsilon>0$ and two distributions $p,q$ on the domain $[n]$ with purified quantum query-access via $U_{p}$ and $U_{q}$ as in \defn{purified-quantum-query}, it takes $\bOt{\frac{\sqrt{n}}{\epsilon}}$ queries to $U_{p},U_{p}^{\dagger},U_{q},U_{q}^{\dagger}$ to decide whether $p\!=\!q$ or $\nrm{p\!-\!q}_{1}\geq\epsilon$ with success probability at least $2/3$.
	Similarly for density operators $\rho,\sigma\in\C^{n\times n}$ with purified quantum query-access via $U_{\rho}$ and $U_{\sigma}$, it takes $\bigO{\frac{n}{\epsilon}}$ queries to $U_{\rho},U_{\rho}^{\dagger},U_{\sigma},U_{\sigma}^{\dagger}$ to decide whether $\rho\!=\!\sigma$ or $\nrm{\rho\!-\!\sigma}_{1}\geq\epsilon$ with success probability at least $2/3$.
\end{corollary}
\begin{proof}
	By the Cauchy-Schwartz inequality we have $\|p-q\|_{2}\geq\frac{1}{\sqrt{n}}\|p-q\|_{1}$, therefore \thm{classical-l2-purified} implies our claim  by taking $\epsilon\leftarrow\epsilon/\sqrt{n}$ therein. Similarly, \thm{quantum-l2-purified} implies our claim for quantum distributions $\rho$ and $\sigma$.
\end{proof}

\subsection{Independence testing with purified query-access}\label{apx:independence}
\begin{corollary}\label{cor:independence-proof}
Given $\epsilon>0$ and a classical distribution $p$ on $\range{n}\times\range{m}$ with the purified quantum query-access via $U_{p}$ as in \defn{purified-quantum-query}, it takes $\bOt{\frac{\sqrt{nm}}{\epsilon}}$ queries to $U_{p},U_{p}^{\dagger}$ to decide whether $p$ is a product distribution on $\range{n}\times\range{m}$ or $p$ is $\epsilon$-far in $\ell^{1}$-norm from any product distribution on $\range{n}\times\range{m}$ with success probability at least $2/3$.
\end{corollary}

\begin{proof}
We define $p_{A}$ to be the margin of $p$ on the first marginal space, i.e., $p_{A}(i)=\sum_{j=1}^{m}p(i,j)$ for all $i\in\range{n}$. We similarly define $p_{B}$ to be the margin of $p$ on the second marginal space, i.e., $p_{B}(j)=\sum_{i=1}^{n}p(i,j)$ for all $j\in\range{m}$. Assume the quantum oracle $U_{p}$ from \defn{purified-quantum-query} acts as
\begin{align*}
U_{p}|0\>_{A}|0\>_{B}|0\>_{C}=\sum_{i=1}^{n}\sum_{j=1}^{m}\sqrt{p(i,j)}|i\>_{A}|j\>_{B}|\psi_{i,j}\>_{C};
\end{align*}
if we denote $|\phi_{i}\>=\sum\limits_{j=1}^{m}\frac{\sqrt{p(i,j)}}{\sqrt{p_{A}(i)}}|j\>|\psi_{i,j}\>$ for all $i\in\range{n}$ and $|\varphi_{j}\>=\sum\limits_{i=1}^{n}\frac{\sqrt{p(i,j)}}{\sqrt{p_{B}(j)}}|i\>|\psi_{i,j}\>$ for all $j\in\range{m}$, then we have
\cvskip{-4mm}
\begin{align*}
U_{p}|0\>_{A}|0\>_{B}|0\>_{C}=\sum_{i=1}^{n}\sqrt{p_{A}(i)}|i\>_{A}|\phi_{i}\>_{B,C}=\sum_{j=1}^{m}\sqrt{p_{B}(j)}|j\>_{B}|\varphi_{j}\>_{A,C}.
\end{align*}
\cvskip{-1mm}\noindent
As a result,
\cvskip{-7mm}
\begin{align*}
(U_{p}\otimes U_{p})(|0\>^{\otimes 6})=\sum_{i=1}^{n}\sum_{j=1}^{m}\sqrt{p_{A}(i)}\sqrt{p_{B}(j)}|i\>|j\>|\phi_{i}\>|\varphi_{j}\>;
\end{align*}
in other words, one purified quantum query to the distribution $p_{A}\times p_{B}$ can be implemented by two queries to $U_{p}$.

If $p$ is a product distribution on $\range{n}\times\range{m}$, then $p=p_{A}\times p_{B}$; if $p$ is $\epsilon$-far in $\ell^{1}$-norm from any product distribution on $\range{n}\times\range{m}$, then $\|p-p_{A}\times p_{B}\|_{1}\geq\epsilon$. Therefore, the problem of independence testing reduces to $\ell^1$-closeness testing for distributions on $\range{n}\times\range{m}$, and hence \cor{independence-proof} follows from \cor{l1-purified}.
\end{proof}

Similarly, \cor{l1-purified} implies that the quantum query complexity of testing independence of quantum distributions is $\bigO{\frac{nm}{\epsilon}}$.

\end{document}